\documentclass[10pt,conference]{IEEEtran}
\IEEEoverridecommandlockouts
\usepackage{times}
\usepackage{soul}
\usepackage{url}
\usepackage[hidelinks]{hyperref}
\usepackage[utf8]{inputenc}
\usepackage[small]{caption}
\usepackage{graphicx}
\usepackage{amsmath}
\usepackage{amsthm}
\usepackage{booktabs}

\usepackage[ruled,linesnumbered]{algorithm2e}
\usepackage[most]{tcolorbox}

\usepackage{xcolor}      
\usepackage{hyperref}    
\usepackage{rotating}  
\usepackage{colortbl}
\usepackage{mathtools}

\usepackage{amssymb}  
\usepackage{threeparttable} 
\usepackage{booktabs}       
\newtheorem{definition}{Definition}

\usepackage{titletoc}
\usepackage{fontawesome5}
\newtheoremstyle{rqplain}
  {1pt}{0pt}{\bfseries}{0pt}{\bfseries}{.}{0.5em}
  {\thmname{#1}\thmnumber{#2}\thmnote{ (#3)}}
\theoremstyle{rqplain}
\newtheorem{researchquestion}{RQ}

\newcommand{\RQicon}{\textcolor{black}{\faIcon{cog}}\,}

\makeatletter
\let\old@researchquestion\researchquestion
\def\researchquestion{%
  \def\@begintheorem##1##2[##3]{%
    \item[\hskip\labelsep\normalfont\bfseries\RQicon ##1\ ##2\if\relax\detokenize{##3}\relax\else\ (##3)\fi.]%
  }%
  \old@researchquestion
}
\makeatother

\theoremstyle{plain}

\usepackage{multirow}
\usepackage{makecell}
\usepackage{bm}
\usepackage{cleveref}
\crefname{figure}{Fig.}{Figs.}
\crefname{table}{Tab.}{Tabs.}
\crefname{equation}{Eq.}{Eqs.}
\crefname{section} {Sec.}{Secs.}
\crefname{chapter} {Ch.}{Chs.}
\crefname{appendix}{App.}{Apps.}

\definecolor{highlight}{HTML}{C8E6C9}
\definecolor{highlight2}{HTML}{DCEDC8}
\definecolor{highlight3}{HTML}{F9FBE7}

\newcommand{\mname}{DST}

\newcommand{\eg}{\emph{e.g.}}

\newtheorem{theorem}{Theorem}
\newtheorem{proposition}{Proposition}

\hypersetup{
    colorlinks=false,          
    citebordercolor={0 1 0},   
    linkbordercolor={1 0 0},   
    urlbordercolor={0 1 1},    
    pdfborder={0 0 2}          
}

\begin{document}

\title{Design-Specification Tiling for ICL-based CAD Code Generation}


\author{
Yali Du,
San-Zhuo Xi,
Hui Sun,
Ming Li\\
National Key Laboratory for Novel Software Technology,
School of Artificial Intelligence, \\Nanjing University, China \protect\\
\{duyl, xisz, sunh, lim\}@lamda.nju.edu.cn%
}


\maketitle

\begin{abstract}
Large language models~(LLMs) have demonstrated remarkable capabilities in code generation, yet their performance remains limited on domain-specific tasks such as Computer-Aided Design~(CAD) code generation, largely due to the scarcity of high-quality training data.
In-Context Learning~(ICL) provides a training-free alternative by prompting LLMs with task-specific exemplars, but its effectiveness critically depends on how exemplars are selected.
Existing selection strategies mainly rely on similarity or point-wise diversity, often overlooking the compositional nature of CAD design specifications, where a query may involve multiple functional requirements, geometric constraints, and design primitives.
As a result, selected exemplars can be individually relevant but collectively redundant, providing insufficient coverage for complex design requirements.
In this work, we propose \emph{knowledge sufficiency} as a principled objective for exemplar selection, aiming to select a compact set of exemplars that maximally satisfies the requirements contained in a target design specification.
To instantiate this objective, we introduce \emph{Design-Specification Tiling~(DST)}, which estimates knowledge sufficiency through a surrogate tiling ratio by decomposing design specifications into multi-granular components and measuring the proportion of query components covered by selected exemplars.
We further show that optimizing this objective can be formulated as a submodular maximization problem, and develop a polynomial-time greedy algorithm tailored to this setting with a $(1-1/e)$-approximation guarantee.
Extensive experiments across multiple LLMs demonstrate that DST substantially improves CAD code generation quality and consistently outperforms existing ICL exemplar selection strategies, highlighting the importance of requirement-level knowledge coverage for domain-specific code generation.
\end{abstract}

\begin{IEEEkeywords}
In-Context Learning, CAD Code Generation, Submodular Optimization
\end{IEEEkeywords}

 \section{Introduction}\label{sec:intro}

Computer-Aided Design~(CAD) serves as a cornerstone of the engineering design process in traditional manufacturing, bridging conceptual designs with precise 3D objects through underlying code that enables model customization and automation.
However, designing complex models in professional CAD software presents substantial challenges, even for highly experienced practitioners.
Consequently, leveraging AI to automatically generate CAD code represents a promising avenue for advancing manufacturing capabilities~\cite{khan2024text2cad,alam2024gencad}.

Recent advances in large language models~(LLMs) have demonstrated remarkable capabilities in code generation~\cite{pmlr-v267-liu25ah,yang2024chain,lyu2025top,wang2025teaching}, motivating their application to CAD code generation~\cite{guan2025cad,niu2025cad}.
However, unlike general-purpose programming languages, CAD code is domain-specific, tightly coupled to specialized software, involves precise dimensions and physical constraints, and is typically hidden behind a GUI.
Most engineers interact with CAD exclusively through visual interfaces, yet the software fundamentally relies on underlying code to generate and modify geometry.
Because CAD code remains largely inaccessible to practitioners and is rarely documented or shared publicly, training corpora for this domain are scarce. 
Therefore, LLMs exhibit inferior performance on CAD code generation tasks, such as Qwen3-30B-A3B, which yields only 15.67\% syntactically valid CAD code for complex tasks in a zero-shot setting~(as shown in \cref{tab:baselines}).

Most existing approaches address this challenge by continually fine-tuning open-source LLMs on manually annotated CAD datasets~\cite{doris2025cad,he2025cad,guan2025cad}.
However, these approaches face fundamental limitations:
fine-tuning demands substantial computational resources, while manual annotation is costly and high-quality proprietary designs remain confidential, ultimately limiting dataset scale and quality.
Recently, In-Context Learning~(ICL) has emerged as a training-free alternative that enhances LLM performance by incorporating task-specific exemplars directly into the input context~\cite{li2024meta,koike2024outfox,yali_tse25_cast}.
This approach enables leveraging state-of-the-art~(SOTA) models~(including open- and closed-source LLMs) without requiring access to training infrastructure or large-scale datasets.

The effectiveness of ICL depends critically on selecting appropriate exemplars.
However, exemplar selection for CAD code generation remains largely unexplored.
Selecting appropriate exemplars for natural language design-specification poses significant challenges.
CAD design specifications typically comprise multiple geometric elements, spatial relationships, and topological constraints, requiring exemplars that provide sufficient information across these multiple dimensions.
Existing ICL methods focus primarily on similarity~\cite{lee2025cropper,DBLP:conf/icml/SinghMDHCS25,DBLP:conf/acl/00010H25} or point-wise diversity~\cite{kapuriya2025exploring,xiao2025role}, often resulting in redundant selections or multi-objective balance, that fail to collectively cover the compositional requirements of complex design specifications.


In this work, we return to first principles of ICL: \emph{exemplar selection should provide sufficient task-relevant knowledge to enable LLMs to solve the current task}.
Motivated by this principle, we propose \textbf{knowledge sufficiency} as a novel objective for exemplar selection that aims to maximally satisfy all requirements within design specifications of the current query.
However,  two practical key challenges arise: 
1)~How to quantify the sufficiency of a group selected exemplars?
2)~How to maximize this sufficiency score with exemplar selection despite it is NP-hard?

To address these challenges, we propose an approach termed \textbf{Design-Specification Tiling~(DST)} that quantifies and maximizes knowledge sufficiency through a surrogate objective.
For the first challenge, DST extracts multi-granular design components from design specifications using n-grams.
DST then tiles the current query's design components using components from selected exemplars.
The proportion of covered components, named \textbf{tiling ratio}, serves as a surrogate measure of knowledge sufficiency.
For the second challenge, we demonstrate that maximizing this surrogate objective constitutes submodular maximization.
While finding the optimal solution is NP-hard, a greedy algorithm provides an efficient solution: starting with an empty set, it iteratively selects exemplars that maximize marginal gain until reaching capacity.
This greedy approach has been proven to achieve a $(1-1/e)$-approximation guarantee in polynomial time, meaning its performance is at least $(1-1/e)$ of optimal.

%
%
%
Our main contributions are summarized as follows:
\begin{itemize}
    \item \emph{\textbf{Knowledge Sufficiency}: A Principled Objective for ICL.} 
    We propose knowledge sufficiency as a novel objective that maximally satisfies design specification requirements within the query, thereby maximizing task-relevant knowledge provided to LLMs.
    
    \item \emph{\textbf{Design-Specification Tiling Ratio}: A Tractable Surrogate Objective.} 
    We quantify knowledge sufficiency via multi-granular design-specification tiling ratio, which extracts n-gram design components and measures the proportion of query components covered by exemplars.
    
    \item \emph{\textbf{Greedy Submodular Maximization}: An efficient Solution with Approximation Guarantees.} 
    We demonstrate this NP-hard optimization problem is submodular and provide a polynomial-time greedy algorithm with a~$(1-1/e)$-approximation guarantee.
    
\end{itemize}

\section{Preliminaries}\label{sec:pre}

In this section, we formulate ICL-based CAD code generation and articulate the key challenges that motivate our approach.

\subsection{ICL-based CAD Code Generation}\label{ssec:formulation}
Given a natural language design specification $X$ describing the CAD model, $X$ typically encompasses multiple aspects of the target model, including geometric elements~(\eg, cylinders, spheres), spatial relationships~(\eg, adjacent, centered), and topological constraints~(\eg, union, intersection).
The goal of CAD code generation is to use an LLM to transform the specification into executable CAD code:~$ X_\text{query} \xrightarrow[]{\text{LLM}} \hat{Y}$, where $\hat{Y}$ realizes the specified design and enables CAD software to render the corresponding digital model.

In ICL, a database $\mathcal{D}=\{(X_i, Y_i)\}_{i=1}^n$ contains $n$ pairs of design specifications and their CAD code.
The database indices are denoted by~$\mathbb{N}_n=\{1,2,\dots, n\}$, and~$S \subseteq \mathbb{N}_n$ represents a subset of indices.
Thus, ICL-based CAD code generation with $k$ exemplars as context can be expressed as:
\begin{equation}\label{eq:icl}
    \mathrm{prompt}[
        \underbrace{X_{S[1]}, Y_{S[1]}, \dots, X_{S[k]}, Y_{S[k]}}_{\text{$S$: context of $k$ paired exemplars}};
         X_{\text{query}}
    ] \xrightarrow[]{\text{LLM}} \hat{Y},
\end{equation}
where $S[i]$ denotes the index of the $i$-th selected exemplar.

\subsection{Motivated Objective: Knowledge Sufficiency}\label{ssec:motivation}
\begin{figure}[!t]
    \centering
    \includegraphics[width=\linewidth]{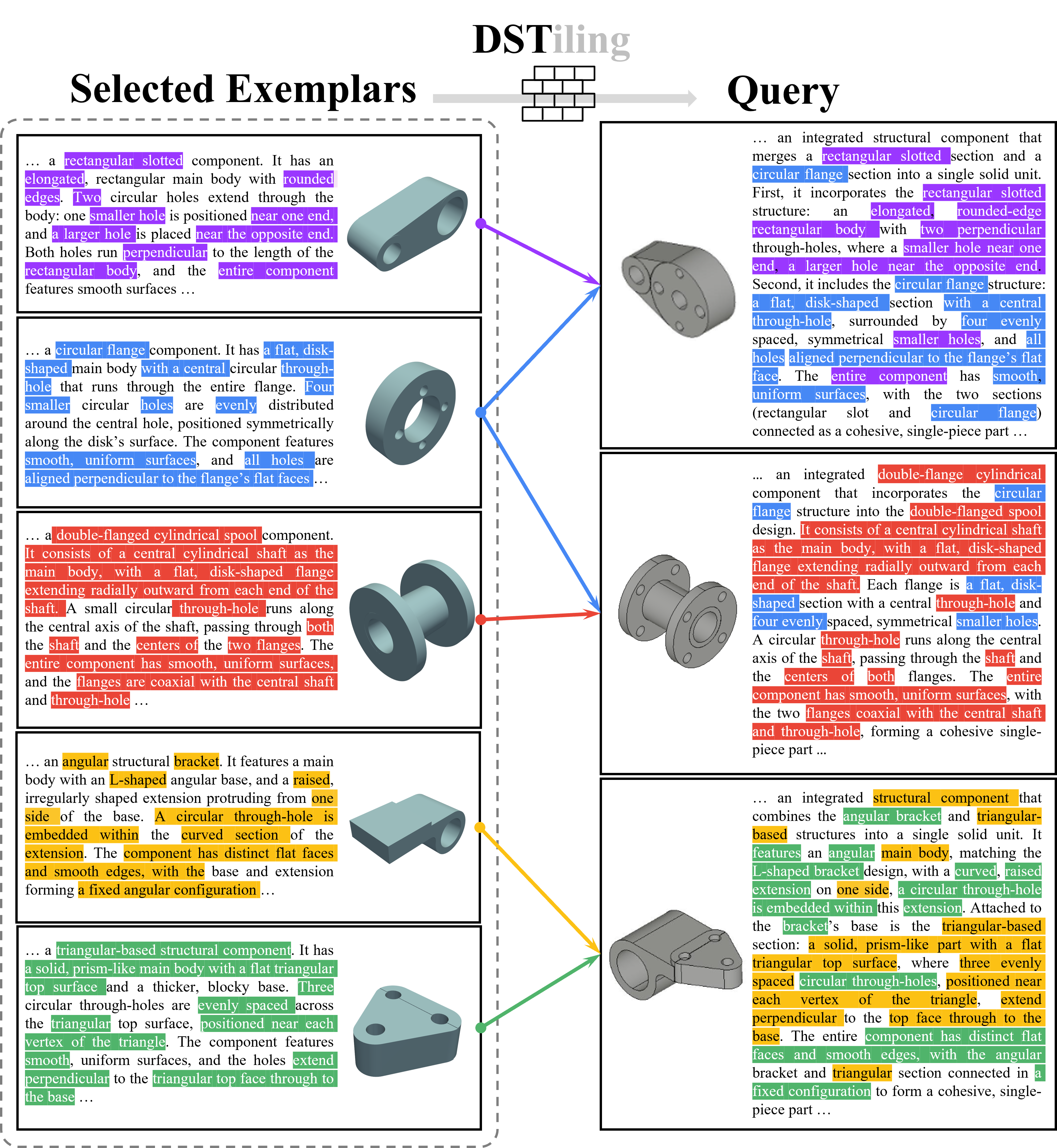}
    \caption{\label{fig:motivation}
        Compositional complexity in CAD design-specifications. The query (right) contains multiple design components (colored highlights), each requiring specific knowledge. Exemplars (left) each cover different component subsets. Effective selection should ensure collective coverage across all components.
    }

\end{figure}

The effectiveness of ICL depends on selecting appropriate exemplars that provide sufficient task-relevant knowledge to the LLM.
However, existing strategies face key limitations when applied to CAD code generation.
CAD design specifications comprise various heterogeneous components~(\eg, geometric elements, spatial relationships, and topological constraints), each requiring different types of knowledge.

\Cref{fig:motivation} illustrates this compositional complexity.
The query~(right) combines multiple distinct design elements: rectangular slotted structures, circular flanges, double-flanged cylinders, angular brackets, and triangular bases.
Mainstream ICL methods select the top-$k$ most similar exemplars, often producing redundant selections that cover overlapping components while missing others~\cite{DBLP:conf/acl/00010H25,li2025large}.
Recent approaches incorporate diversity constraints but optimize for point-wise diversity without considering collective coverage of the query's compositional requirements~\cite{kapuriya2025exploring,xiao2025role}.

We return to the first principle of ICL: \emph{providing sufficient task-relevant knowledge to enable the LLM to solve the current task}.
We formalize this through knowledge sufficiency: selecting exemplars that collectively cover all components of the query.
Formally, we seek an exemplar set $S^*$ of size $k$ that maximizes:
\begin{equation}\label{eq:objective}
    S^* = \mathrm{argmax}_{S \subseteq \mathbb{N}_n, |S|=k} f_\text{suff}(S; X_\text{query}),
\end{equation}
where $f_\text{suff}(S; X_\text{query})$ quantifies the knowledge sufficiency provided by $S$ for query $X_\text{query}$.

This objective presents two practical challenges: (1)~How to measure knowledge sufficiency? (2)~How to efficiently search the exponential space of $\binom{n}{k}$ possible sets?
These challenges motivate \mname{}~(\Cref{sec:method}), which quantifies sufficiency through multi-granular component tiling and leverages submodular optimization to find an efficient solution.

\section{Method: Design-Specification Tiling~(DST)}\label{sec:method}
To operationalize the ideal knowledge sufficiency objective, we propose Design-Specification Tiling~(DST), which addresses both the quantification and optimization challenges through a surrogate objective.
The core insight of DST is to measure knowledge sufficiency through the \emph{tiling ratio}: 
decomposing design specifications into multi-granular components and quantifying the proportion of query components that are tiled by selected exemplars.
This tiling ratio serves as a computable surrogate for knowledge sufficiency, a higher ratio indicates that exemplars collectively provide more comprehensive knowledge about the query's requirements.

DST operates in three stages:
\textbf{(1)}~Extract multi-granular design-specification components from natural language using n-grams~(\cref{sec:design-specification});
\textbf{(2)}~Formulate the tiling ratio as a surrogate objective that measures how well exemplar components tile the query components~(\cref{sec:tiling-objective});
\textbf{(3)}~Optimize this objective efficiently via submodular maximization~(\cref{sec:optimization}).

\subsection{Multi-Granular Specification Components}\label{sec:design-specification}

To compute the tiling ratio, we introduce a structured representation that decomposes design specifications into \emph{multi-granular components}.
This representation enables quantifying the semantic overlap between queries and exemplars at multiple levels of granularity, from atomic design elements to complex compositional patterns.

Given a natural language design specification $X$, we extract components using n-grams of varying sizes.
An n-gram of size $n$ captures a consecutive sequence of $n$ words, representing a design element at a specific semantic granularity.
Unigrams such as ``cylinder'' capture atomic geometric primitives, while trigrams such as ``cylinder with holes'' capture compositional relationships between elements.
This multi-granular representation is critical for CAD specifications, which inherently describe designs at multiple levels, from individual geometric shapes to their spatial arrangements and topological operations.
Formally, we apply a sliding window of size $n$ to extract all consecutive n-grams from specification $X$.
The set of n-grams of size $n$ is defined as:
\begin{equation}
    \mathcal{C}^{(n)}(X) = \{X[i:i+n] \mid i = 1, 2, \ldots, L-n+1\}\,,
\end{equation}
where $X[i:i+n]$ denotes the n-gram starting at position $i$.

Moreover, to capture design knowledge across semantic levels, we extract n-grams for exponentially spaced sizes $N = \{2, 4, 8, 16, 32\}$. \

The complete component set is:
\begin{equation}\label{eq:components}
    \mathcal{C}(X) = \bigcup_{n\in N} \mathcal{C}^{(n)}(X)\,.
\end{equation}
For $i$-th exemplar with specification $X_i$, we denote its component set as $\mathcal{C}_i = \mathcal{C}(X_i)$.
For query $X_{\text{query}}$, we denote its component set as $\mathcal{C}_{\text{query}} = \mathcal{C}(X_{\text{query}})$.
This exponential spacing balances expressiveness with computational efficiency, capturing design elements from local word pairs to longer compositional phrases while avoiding exhaustive enumeration of all intermediate sizes. \

\subsection{Tiling Ratio}\label{sec:tiling-objective}

Having extracted multi-granular components, we now formulate the tiling ratio as a tractable surrogate for knowledge sufficiency.
The key idea is to measure what proportion of the query's components are tiled by components from selected exemplars.
For an exemplar set $S \subseteq \mathbb{N}_n$, the union of components from all selected exemplars is:
\begin{equation}
    \mathcal{C}(S) = \bigcup_{i \in S} \mathcal{C}_i.
\end{equation}
The components that tile the query are those shared between the exemplar set and the query: $\mathcal{C}(S) \cap \mathcal{C}_{\text{query}}$.
The tiling ratio quantifies the proportion of query components that are tiled:
\begin{equation}\label{eq:tiling-ratio}
    f_\text{suff}(S; X_{\text{query}}) :=  \frac{w(\mathcal{C}(S) \cap \mathcal{C}_{\text{query}})}{w(\mathcal{C}_{\text{query}})}\,.
\end{equation}
Here $w(\mathcal{C}(X))=\sum_{n\in N}n\cdot|\mathcal{C}^{(n)}(X)|$ is the weighted tiling of components in $X$, where longer specification components receive proportionally higher weight, thus encoding the disparate impact of multi-granularity specifications.
This ratio ranges from 0~(no query components are tiled) to 1.0~(all query components are tiled).
A higher tiling ratio indicates that the selected exemplars collectively provide knowledge about more aspects of the query.

Our goal is to select $k$ exemplars that maximize the tiling ratio.
Since $|\mathcal{C}_{\text{query}}|$ is constant for a given query, maximizing the tiling ratio is equivalent to maximizing the number of tiled query components:
\begin{equation}\label{eq:tiling-objective}
    S^* = \arg\max_{S \subseteq \mathbb{N}_n, |S|=k} w(\mathcal{C}(S) \cap \mathcal{C}_{\text{query}})\,,
\end{equation}
where $k$ is the exemplar set capacity, a hyperparameter that controls the number of selected exemplars.

This formulation follows our knowledge sufficiency principle: we seek exemplars whose components, when combined, tile the maximum portion of the query's compositional structure.
However, this optimization involves searching an exponential space of $\binom{n}{k}$ possible exemplar combinations.

\begin{figure*}[!t]
    \centering
    \includegraphics[width=\linewidth]{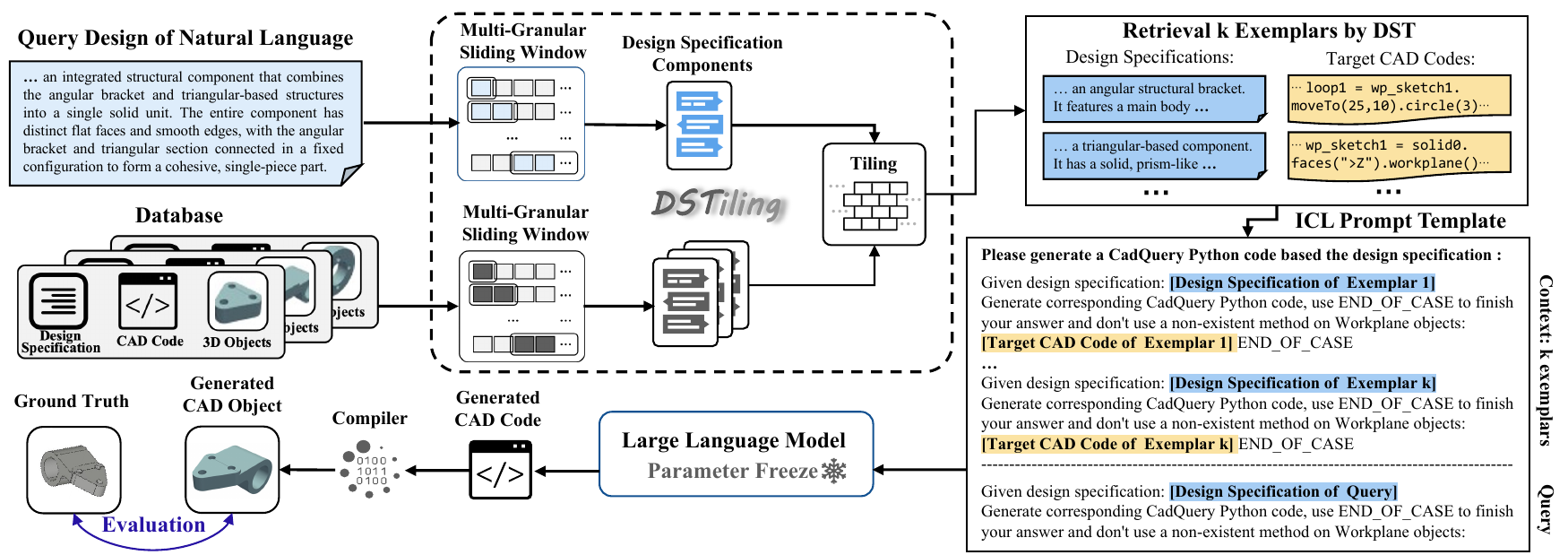}
    \caption{Overview of the Design-Specification Tiling~(DST) framework. Design specification components are extracted from the query and database using multi-granular sliding windows~(left). The DST retrieval algorithm selects $k$ exemplars that maximize coverage of query components through greedy optimization~(center). Selected exemplars are formatted into an ICL prompt template to guide the large language model in generating CAD code, which is then compiled and evaluated~(right).}
    \label{fig:pipeline}
\end{figure*}

\subsection{Efficient Submodular Maximization}\label{sec:optimization}

The optimization problem in \cref{eq:tiling-objective} requires selecting $k$ exemplars from an exponential search space of $\binom{n}{k}$ possibilities, which is computationally intractable for exhaustive search.
We address this challenge by leveraging the submodular structure of the tiling objective, which enables efficient approximation with provable guarantees.

We first establish the key structural property of our objective function through the concept of submodularity.
\begin{definition}[Submodularity~\cite{nemhauser1978best}]
A set function $f: 2^{\Omega} \rightarrow \mathbb{R}_{\geq 0}$ is submodular if for all $A \subseteq B \subseteq \Omega$ and $x \in \Omega \setminus B$:
\begin{equation}\label{eq:submod-condition}
    f(A \cup \{x\}) - f(A) \geq f(B \cup \{x\}) - f(B)\,,
\end{equation}
and $f$ is monotone if $f(A) \leq f(B)$ for all $A \subseteq B \subseteq \Omega$.
\end{definition}

Intuitively, submodularity captures the diminishing returns property:
adding an exemplar to a smaller set provides at least as much benefit as adding it to a larger set, since the larger set likely already contains overlapping components.

\begin{proposition}\label{prop:submodular}
The tiling objective $f(S) = |\mathcal{C}(S) \cap \mathcal{C}_{\text{query}}|$ satisfies the following properties:
\begin{itemize}
    \item \textit{Non-negativity}: $f(S) \geq 0$ for all $S \subseteq \mathbb{N}_n$;
    \item \textit{Monotonicity}: $f(A) \leq f(B)$ for all $A \subseteq B$;
    \item \textit{Submodularity}: For all $A \subseteq B$ and $x \notin B$, $f(A \cup \{x\}) - f(A) \geq f(B \cup \{x\}) - f(B)$.
\end{itemize}
\end{proposition}
\begin{proof}\label{proof:submodular}
We first establish notation consistent with the main text. Let $\mathcal{C}_{\text{query}}$ denote the weighted multi-granular component set of the query $X_{\text{query}}$. For an exemplar subset $S \subseteq \mathbb{N}_n$, define 
\begin{equation}
\mathcal{C}(S) = \bigcup_{i\in S}\mathcal{C}_i
\end{equation}
as the union of components from all exemplars in $S$. The tiling objective (\cref{eq:tiling-objective}) is
\begin{equation}
    f(S) = w\left(\mathcal{C}(S) \cap \mathcal{C}_{\text{query}}\right),
\end{equation}
where the weighted size function 
\begin{equation}
w(\mathcal{C}(\cdot)) = \sum_{n\in \mathcal{N}} n \cdot |\mathcal{C}^{(n)}(\cdot)|
\end{equation}
assigns weight $n$ to each $n$-gram component. For convenience, define
\begin{equation}
\mathcal{T}(S) = \mathcal{C}(S) \cap \mathcal{C}_{\text{query}},
\end{equation}
so that $f(S) = w(\mathcal{T}(S))$.

To complete the proof, we establish that $f(S)$ satisfies non-negativity, monotonicity, and submodularity.

\noindent\textbf{Non-negativity.}
The weighted size function $w(\cdot)$ sums non-negative terms: each weight $n \in \mathcal{N}$ is positive and each cardinality $|\mathcal{C}^{(n)}(\cdot)|$ is non-negative. Therefore, $w(\mathcal{T}(S)) \geq 0$ for all $S \subseteq \mathbb{N}_n$, which implies $f(S) \geq 0$.

\noindent\textbf{Monotonicity.}
Consider arbitrary subsets $A, B \subseteq \mathbb{N}_n$ with $A \subseteq B$. By monotonicity of set union,
\begin{equation}
    \mathcal{C}(A) = \bigcup_{i\in A}\mathcal{C}_i \subseteq \bigcup_{i\in B}\mathcal{C}_i = \mathcal{C}(B).
\end{equation}
Intersecting with $\mathcal{C}_{\text{query}}$ preserves inclusion:
\begin{equation}
    \mathcal{T}(A) \subseteq \mathcal{T}(B).
\end{equation}
Since $w(\cdot)$ is monotone with respect to set inclusion, 
\begin{equation}
    f(A) = w(\mathcal{T}(A)) \leq w(\mathcal{T}(B)) = f(B).
\end{equation}
This establishes monotonicity.

\noindent\textbf{Submodularity.}
Define the marginal gain
\begin{equation}
    \Delta_A(x) := f(A\cup\{x\}) - f(A).
\end{equation}
Observe that
\begin{equation}
\begin{aligned}
    \mathcal{T}(A\cup\{x\}) &= \mathcal{C}(A\cup\{x\}) \cap \mathcal{C}_{\text{query}} \\
    &= \left(\mathcal{C}(A) \cup \mathcal{C}_x\right) \cap \mathcal{C}_{\text{query}} \\
    &= \mathcal{T}(A) \cup \left(\mathcal{C}_x \cap \mathcal{C}_{\text{query}}\right),
\end{aligned}
\end{equation}
where $\mathcal{C}_x$ is the component set of exemplar $x$. Since $w(\cdot)$ is additive over disjoint unions, the marginal gain is
\begin{equation}
    \Delta_A(x) = w\left(\mathcal{T}(A\cup\{x\}) \setminus \mathcal{T}(A)\right).
\end{equation}
This simplifies to
\begin{equation}
    \Delta_A(x) = w\left(\left(\mathcal{C}_x \cap \mathcal{C}_{\text{query}}\right) \setminus \mathcal{T}(A)\right).
\end{equation}
Similarly,
\begin{equation}
    \Delta_B(x) = w\left(\left(\mathcal{C}_x \cap \mathcal{C}_{\text{query}}\right) \setminus \mathcal{T}(B)\right).
\end{equation}

Since $A \subseteq B$, monotonicity yields $\mathcal{T}(A) \subseteq \mathcal{T}(B)$. For any fixed set, removing a larger set yields a smaller remainder:
\begin{equation}
    \left(\mathcal{C}_x \cap \mathcal{C}_{\text{query}}\right) \setminus \mathcal{T}(A) \supseteq \left(\mathcal{C}_x \cap \mathcal{C}_{\text{query}}\right) \setminus \mathcal{T}(B).
\end{equation}
Applying monotonicity of $w(\cdot)$ gives
\begin{equation}
    \Delta_A(x) \geq \Delta_B(x),
\end{equation}
verifying \eqref{eq:submod-condition}. This establishes submodularity.

All three properties are now proven.
\end{proof}

These properties enable the application of efficient greedy algorithms with strong theoretical guarantees.
For non-negative, monotone submodular functions, a greedy selection strategy achieves near-optimal performance~\cite{sunh_mdp3,yali_tse25_cast}, as formalized in the following theorem.

\begin{theorem}[$(1-1/e)$-Approximation~\cite{nemhauser1978best}]\label{thm:approximation}
For a non-negative, monotone submodular function $f$, the greedy algorithm that iteratively selects elements maximizing marginal gain produces a solution $\hat{S}$ of size $k$ satisfying:
\begin{align}
        f&(\hat{S}) 
            \geq \left( 1 - \left(1 - \frac{1}{k}\right)^k \right) f(S^*) \\
            &\geq \lim_{k \to +\infty} \left( 1 - \left(1 - \frac{1}{k}\right)^k \right) f(S^*) 
            = \left(1 - \frac{1}{e}\right) f(S^*)\,, \notag
\end{align}
where $S^* = \arg\max_{S \subseteq \Omega, |S| \leq k} f(S)$ is the optimal solution.
\end{theorem}

Based on this theorem, our algorithm starts with an empty set $S = \emptyset$ and iteratively selects the exemplar with maximum marginal gain:
\begin{equation}\label{eq:greedy-step}
     i^* = \mathop{\mathrm{argmax}}_{i \in \mathbb{N}_n \setminus S} \left[w(\mathcal{C}(S \cup \{i\}) \cap \mathcal{C}_{\text{query}}) - w(\mathcal{C}(S) \cap \mathcal{C}_{\text{query}})\right]
\end{equation}
adding $i^*$ to $S$ in each iteration.
This continues until $|S| = k$ or no exemplar provides positive marginal gain.

The greedy algorithm runs in polynomial time, with complexity depending on the efficiency of the component set operations.
Importantly, the $(1-1/e) \approx 0.63$ approximation guarantee ensures that our solution achieves at least 63\% of the optimal tiling ratio.
For monotone submodular maximization under a cardinality constraint $k$, a tighter bound of $1-(1-1/k)^k$ can be derived.
For typical values such as $k=3$, this yields~$1 - (2/3)^3 \approx 70.4\%$.

\subsection{Overall Framework}\label{sec:framework}
\Cref{fig:pipeline} shows the complete DST framework, which operates in two main stages: exemplar retrieval and code generation.

In the retrieval stage, given a natural language query $X_{\text{query}}$, we first extract its design specification components $\mathcal{C}_{\text{query}}$ using multi-granular sliding windows~(\cref{eq:components}).
We then apply our DST algorithm to select $k$ exemplars $S^* \subseteq \mathbb{N}_n$ that maximize the tiling ratio $f_\text{TR}(S^*, X_{\text{query}})$ over the database.
This greedy selection process leverages pre-computed component sets $\{\mathcal{C}_i\}_{i=1}^n$ for all database exemplars to efficiently identify the most complementary examples.

In the generation stage, the selected exemplars and their corresponding CAD code are organized into an ICL prompt, with each exemplar formatted as a specification-code pair.
The query's specification is then appended to this prompt.
A LLM with frozen parameters processes this prompt to generate the target CAD code, which is subsequently compiled and executed to render the final digital model.
The pseudocode of the algorithm is detailed in Algorithm~\ref{algo:dst}.

\begin{algorithm}[t]
    \KwIn{Query components $\mathcal{C}_{q}$, exemplar components $\{\mathcal{C}_i\}_{i=1}^n$, capacity $k$.}
    \KwOut{Selected exemplar set $S^* \subseteq \mathbb{N}_n$.}

    $S \leftarrow \emptyset$, $\mathcal{T} \leftarrow \emptyset$\tcp*{\small covered query components}

    \For{$t \leftarrow 1$ \KwTo $k$}{
        $s \leftarrow \arg\max_{j \in \mathbb{N}_n \setminus S}
        \left[
        w\bigl(\mathcal{T} \cup (\mathcal{C}_j \cap \mathcal{C}_{q})\bigr) - w(\mathcal{T})
        \right]$\tcp*{\small largest marginal tiling gain}

        \If{$w\bigl(\mathcal{T} \cup (\mathcal{C}_s \cap \mathcal{C}_{q})\bigr) = w(\mathcal{T})$}{
            \textbf{break}\tcp*{\small no new coverage}
        }

        $\mathcal{T} \leftarrow \mathcal{T} \cup (\mathcal{C}_s \cap \mathcal{C}_{q})$\;
        $S \leftarrow S \cup \{s\}$\;
    }

    \Return{$S^* = S$}\;

    \caption{Pseudocode of DST Algorithm}
    \label{algo:dst}
\end{algorithm}

\section{Experiment}
We conduct extensive experiments to demonstrate the effectiveness of the proposed method.

\subsection{Dataset and Baselines}
\textbf{Dataset.}
We adopt the widely used Text2CAD dataset~\cite{khan2024text2cad}, including 178K samples. After filtering entries with missing inputs and removing duplicates, it retains 151K valid samples.
We quantify the complexity of each instance using three complementary metrics:
\begin{itemize}
\item $\text{Len}(NL_i)$: the linguistic complexity, defined as the length of the natural language design specification.
\item $\text{Geom}(CAD_i)$: the geometric complexity~\cite{Alrashedycad}, defined as the total number of edges and faces in the CAD model.
\item $\text{Ops}(Code_i)$: the operational complexity, measured by the number of valid operations of CadQuery code.
\end{itemize}
To ensure equal contributions from each dimension and to mitigate biases arising from differing scales, all three metrics were standardized using min-max normalization.
Then a comprehensive complexity score, denoted as $\textbf{Complexity}_{i}$, is construct by aggregating these three normalized metrics as shown in \cref{Complexity}. 
Finally, samples are ranked by complexity and partitioned into three equal-sized tiers~({Easy}, {Middle}, and {Hard}). 
For the final experimental setup, 300 samples were randomly selected from each tier as the test set, while the remaining samples served as the ICL database.
\begin{equation}\label{Complexity}
\begin{aligned}
    \textbf{N}(x_i) &= \frac{x_i - \min(\mathbf{x})}{\max(\mathbf{x}) - \min(\mathbf{x})}\,, \\
    \text{Complexity}_i &= \textbf{N}(\text{Len}(NL_i) +  \textbf{N}(\text{Geom}(CAD_i)) \\ &+ \textbf{N}(\text{Ops}(Code_i)))\,.
\end{aligned}
\end{equation}

To evaluate LLMs on CAD code generation tasks of varying complexity, we rank samples using above intrinsic metrics, and divide them into three subsets: \textbf{Easy}, \textbf{Middle}, and \textbf{Hard}. 
From each subset, we randomly select 300 samples as the test set, while using the remaining samples as the in-context learning database.
As shown in \cref{tab:data_statistics}, from Easy to Hard, the three subsets exhibit clear increases across linguistic, geometric, and operational complexity, validating the effectiveness of our stratification strategy and providing a robust basis for fine-grained performance analysis.
\begin{table}[!h]
    \centering
    \small
    \setlength{\tabcolsep}{1.6pt} 
    \caption{Dataset statistics across complexity groups. For each group, the top row shows the \textbf{Mean} and the bottom row shows the \textbf{Median}.}
    \label{tab:data_statistics}
    \begin{tabular}{llccc|c} 
        \toprule
        \textbf{Group} & \textbf{stat.} & {$\text{Len}(NL_i)$} & {$\text{Geom}(CAD_i)$} & {$\text{Ops}(Code_i)$} & {Test Set} \\ 
        \midrule
        \multirow{2}{*}{\textbf{Easy}}   & Mean   & 147.67 & 17.33 & 5.67 & \multirow{2}{*}{300} \\
                                         & Median & 143.00 & 18.00 & 5.00 & \\
        \midrule
        \multirow{2}{*}{\textbf{Middle}} & Mean   & 195.00 & 27.67 & 9.67 & \multirow{2}{*}{300} \\
                                         & Median & 182.00 & 22.00 & 9.00 & \\
        \midrule
        \multirow{2}{*}{\textbf{Hard}}   & Mean   & 327.67 & 51.33 & 21.33 & \multirow{2}{*}{300} \\
                                         & Median & 301.50 & 50.00 & 19.00 & \\
        \bottomrule
    \end{tabular}
\end{table}

\textbf{Baselines.}
We compare against SOTA LLMs, including the open-source Qwen3~\cite{yang2025qwen3}, DeepSeek-V3~\cite{liu2024deepseek}, and the closed-source Claude 4.5-Haiku~\cite{anthropic2025claude}.
We benchmark the proposed DST strategy against four mainstream ICL exemplar selection baselines~\cite{bm25llm,xiao2025role,yali_tse25_cast}. The ICL exemplar selection task is formalized as choosing a size-$k$ subset $S$ to maximize a scoring function w.r.t.\ the query code $X_{\text{query}}$.
\begin{itemize}
    \item \textbf{Random Sampling}: Randomly selects $k$ exemplars from the candidate pool without leveraging query information.
    \item \textbf{LDSIM}: Uses Levenshtein (edit) distance to measure code similarity; selects exemplars by maximizing the negative sum of edit distances to the query. It is lightweight and widely used in ICL retrieval.
    \item \textbf{BM25}: A classic TF-IDF enhanced retrieval model that ranks candidates by textual relevance to the query, serving as a standard retrieval baseline.
    \item \textbf{Diversity (Clustering-based)}: Applies k-means clustering on CodeBERT embeddings to split candidates into $k$ clusters, then picks the centroid-closest sample from each cluster via cosine similarity for diverse exemplar selection.
\end{itemize}

\subsection{Evaluation Details}
\label{app:metric}

\textbf{Evaluation Metrics.}
Following prior CAD generation studies~\cite{guan2025cad,doris2025cad,he2025cad,niu2025cad,Alrashedycad}, we evaluate the generated CAD models using four widely-used metrics.
\begin{itemize}
    \item \textbf{Valid Syntax Rate (VSR)} evaluates code-level correctness by measuring the percentage of generated CadQuery scripts that can be executed as Python files without syntax or runtime errors.

    \item \textbf{Intersection over Union (IoU)} measures the volumetric overlap between the generated model $P$ and the ground truth model $Q$:
    \begin{equation}
        \operatorname{IoU}(P,Q) = \frac{|P \cap Q|}{|P \cup Q|}.
        \label{eq:iou}
    \end{equation}

    \item \textbf{Chamfer Distance (CD)} measures the bidirectional nearest-neighbor distance between point clouds sampled from the generated model and the ground truth model.

    \item \textbf{Edge Chamfer Distance (ECD)} is a structure-aware variant of CD computed on edge point clouds, making it more sensitive to geometric details such as corners, contours, and boundaries.
\end{itemize}

We utilize the Python-based CadQuery library\footnote{https://cadquery.readthedocs.io/} as the CAD compiler for rendering 3D objects from generated code.


\textbf{Handling of Invalid Outputs.}
To ensure a rigorous and impartial evaluation across all metrics (IoU, CD, and ECD ) we adopt a worst-case penalty for samples where model generation fails or errors occur. Specifically, these invalid outputs are represented as a singleton point set located at the coordinate origin $(0,0,0)$. Under this convention, the IoU is consistently recorded as 0, while CD and ECD values are calculated based on the distance between the origin point and the ground truth model. Since all target shapes undergo a normalized transformation process prior to evaluation, these distance-based penalties remain bounded and do not exert a disproportionate influence on the final performance results.


\textbf{Coordinate Alignment} \label{app:alignment}
Due to the inherent ambiguity in natural language descriptions regarding reference frames and absolute coordinates, discrepancies in relative position and orientation may exist between the generated solid model and the ground truth shape, denoted as $\textbf{P}$ and $\textbf{Q}$. Furthermore, variations in absolute scale across samples can lead to inconsistent weighting in distance-based metrics. To eliminate these shape-independent biases, we apply a series of rigid transformations (including translation, rotation, and scaling) to normalize and align $\textbf{P}$ with $\textbf{Q}$ before metric computation.


\paragraph{Normalization}
First, under the assumption of uniform mass distribution, we translate the centroid of the ground truth model $\textbf{Q}$ to the coordinate origin. For scaling, we utilize the \textbf{radius of gyration} \cite{scale_radius,doris2025cad} as a dimensional reference. Specifically, we apply uniform scaling such that the mean distance from the model's differential mass elements to the origin is normalized to a unit length.\ 
The detailed operations are as follows:
\begin{equation}
\begin{aligned}
    \text{Centroid}(\mathbf{Q}) & = \bar{\mathbf{q}} = \frac{1}{|\mathbf{Q}|} \int_{\mathbf{x} \in \mathbf{Q}} \mathbf{x} \, d\mathbf{x} \\
    \text{Translate}(\mathbf{Q}, -\bar{\mathbf{q}}) &= \{ \mathbf{x} - \bar{\mathbf{q}} \mid \mathbf{x} \in \mathbf{Q} \} \\
    R_g(\mathbf{Q}) &= \sqrt{\frac{\text{Tr}(\mathbf{I}_{\mathbf{Q}})}{2 \cdot \text{Vol}(\mathbf{Q})}} \\
    \text{Scale}(\mathbf{Q},R_g(\mathbf{Q})) &= \{ \frac{\mathbf{x}}{R_g(\mathbf{Q})} \mid \mathbf{x} \in \mathbf{Q} \}
\end{aligned}
\end{equation}
where $\mathbf{I}_{\mathbf{Q}}$ stands for the inertia matrix of $\textbf{Q}$ and $\text{Vol}(\mathbf{Q})$ stands for the volume of $\textbf{Q}$.


\begin{figure}[t]
    \centering
        \includegraphics[width=\linewidth]{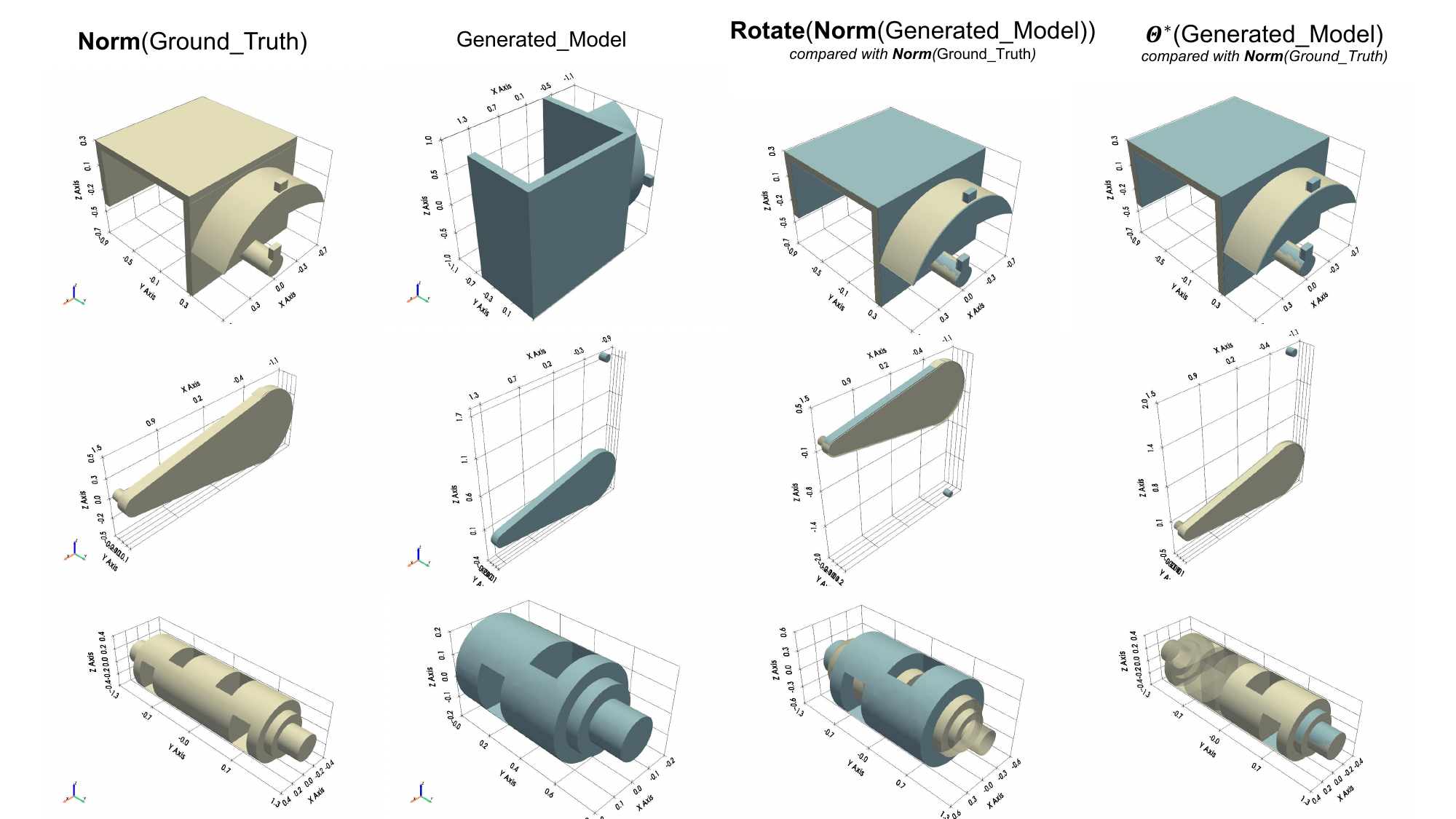}
    \caption{The normalized ground truth $\textbf{Norm}(\textbf{Q})$ and the generated models $\textbf{P}$ under different transformations: it is observed that when the generated model is substantially consistent with the ground truth, normalization and orthogonal rotation can effectively achieve the optimal alignment. However, in cases where the generated model is only partially correct, a more exhaustive search across a broader range of transformation combinations is necessary.}
    \label{fig:aligned_res}
\end{figure}

\begin{table*}[ht]
    \centering
    \small
    \setlength\tabcolsep{5.2pt}
    \caption{\label{tab:baselines}
    	Performance of LLMs with various ICL exemplar strategies (5-shot).
	}
    \resizebox{\textwidth}{!}{
    \begin{threeparttable}
    \begin{tabular}{l|lcccccccccccc}
       \toprule
       \multirow{2}{*}{} & \multirow{2}{*}{\textbf{Strategy}} & 
       \multicolumn{4}{c}{\textbf{Easy}} & \multicolumn{4}{c}{\textbf{Middle}} & \multicolumn{4}{c}{\textbf{Hard}} \\ 
       \cmidrule(lr){3-6}  \cmidrule(lr){7-10}  \cmidrule(lr){11-14}  
       & & \textbf{VSR}\textcolor{green}{$\bm{\uparrow}$} & \textbf{IoU}\textcolor{green}{$\bm{\uparrow}$} & \textbf{CD}\textcolor{red}{$\bm{\downarrow}$}  & \textbf{ECD}\textcolor{red}{$\bm{\downarrow}$} & \textbf{VSR}\textcolor{green}{$\bm{\uparrow}$} &\textbf{IoU}\textcolor{green}{$\bm{\uparrow}$} & \textbf{CD}\textcolor{red}{$\bm{\downarrow}$}  & \textbf{ECD}\textcolor{red}{$\bm{\downarrow}$} & \textbf{VSR}\textcolor{green}{$\bm{\uparrow}$} & \textbf{IoU}\textcolor{green}{$\bm{\uparrow}$} & \textbf{CD}\textcolor{red}{$\bm{\downarrow}$}  & \textbf{ECD}\textcolor{red}{$\bm{\downarrow}$} \\
      \midrule
       \multirow{7}{*}{\rotatebox{90}{Qwen3-30B-A3B}}
           & Zero-Shot & 0.5267 & 0.2609 & 0.6133 & 0.4877 & 0.2933 & 0.1274 & 0.6788 & 0.5354 & 0.1567 & 0.0738 & 0.6587 & 0.5593 \\
           & Random    & 0.7400 & 0.5056 & 0.3978 & 0.3481 & 0.6467 & 0.3942 & 0.4294 & 0.3550 & 0.4133 & 0.2228 & 0.4889 & 0.4393 \\  
           & LDSIM     & \textbf{0.8233} & \underline{0.5260} & \underline{0.3607} & \underline{0.3330} & \underline{0.6733} & 0.4200&	\underline{0.3964} & \underline{0.3391} & 0.5033 & 0.2660 & 0.4610 & 0.4068 \\
           & Diversity & 0.7867 & 0.5136 & 0.3925 & 0.3479 & 0.6467 & 0.3741 & 0.4435 & 0.3742 & 0.6067 & 0.3216 & 0.4067 & 0.3690 \\
           & BM25 & 0.6733 & 0.4982 & 0.4144 & 0.3670 & 0.6567 & \underline{0.4228} & 0.4050 & 0.3416 & \underline{0.6767} & \underline{0.3621} & \underline{0.3867} & \underline{0.3535} \\
           &\cellcolor{blue!10}\textbf{\mname{}~(Ours)}
                       & \cellcolor{blue!10}0.8067 & \cellcolor{blue!10}\textbf{0.6225} &  \cellcolor{blue!10}\textbf{0.3021} & \cellcolor{blue!10}\textbf{0.2980} & \cellcolor{blue!10}\textbf{0.8133} & \cellcolor{blue!10}\textbf{0.5613} & \cellcolor{blue!10}\textbf{0.2760} & \cellcolor{blue!10}\textbf{0.2511} & \cellcolor{blue!10}\textbf{0.6833} & \cellcolor{blue!10}\textbf{0.4260} & \cellcolor{blue!10}\textbf{0.3301} & \cellcolor{blue!10}\textbf{0.3044} \\
          &\cellcolor{blue!10}\emph{Improvement(\%)} 
                        & \cellcolor{blue!10}\textcolor{red!80!black}{-2.02} & \cellcolor{blue!10}\textcolor{green!60!black}{+18.35} & \cellcolor{blue!10}\textcolor{green!60!black}{+16.25} & \cellcolor{blue!10}\textcolor{green!60!black}{+10.51} & \cellcolor{blue!10}\textcolor{green!60!black}{+20.79} & \cellcolor{blue!10}\textcolor{green!60!black}{+33.06} & \cellcolor{blue!10}\textcolor{green!60!black}{+30.37} & \cellcolor{blue!10}\textcolor{green!60!black}{+25.95} & \cellcolor{blue!10}\textcolor{green!60!black}{+0.97} & \cellcolor{blue!10}\textcolor{green!60!black}{+17.65} & \cellcolor{blue!10}\textcolor{green!60!black}{+14.64} & \cellcolor{blue!10}\textcolor{green!60!black}{+13.89} \\ 
       \midrule
        \multirow{7}{*}{\rotatebox{90}{DeepSeek-V3}}
           & Zero-Shot & 0.4900 & 0.2524 & 0.6709 & 0.5184 & 0.5000 & 0.2309 & 0.5472 & 0.4338 & 0.3833 & 0.1769 & 0.5633 & 0.4742 \\
           & Random &    0.9467 & 0.7694 & 0.1041 & 0.1769 & 0.9100 & 0.6829 & 0.1502 & 0.1874 & 0.7900 & 0.4947 & 0.2516 & 0.2524 \\  
           & LDSIM  &    \underline{0.9667} & \underline{0.8003} & 0.0907 & \underline{0.1682} & \underline{0.9267} & 0.6985 & 0.1440 & \underline{0.1782} & \textbf{0.9067} & 0.5570 & \underline{0.1913} & \underline{0.2086} \\
           & Diversity & \textbf{0.9767} & 0.7959 & \underline{0.0818} & 0.1612 & 0.9067 & 0.6715 & 0.1566 & 0.1840 & 0.8867 & 0.5263 & 0.2139 & 0.2235\\
           & BM25 &  0.9167&	0.7687&0.1193&0.1932&0.9133& \underline{0.7051} & \underline{0.1439} &0.1849&  0.8700 & \underline{0.5586} &	0.2009	&0.2123  \\
           &\cellcolor{blue!10}\textbf{\mname{}~(Ours)} 
                    & \cellcolor{blue!10}\textbf{0.9767} & \cellcolor{blue!10}\textbf{0.8227} & \cellcolor{blue!10}\textbf{0.0778} & \cellcolor{blue!10}\textbf{0.1580} & \cellcolor{blue!10}\textbf{0.9367} & \cellcolor{blue!10}\textbf{0.7083} & \cellcolor{blue!10}\textbf{0.1229} & \cellcolor{blue!10}\textbf{0.1646}  & \cellcolor{blue!10}\underline{0.8900} & \cellcolor{blue!10}\textbf{0.5905} & \cellcolor{blue!10}\textbf{0.1807} & \cellcolor{blue!10}\textbf{0.1988} \\
           &\cellcolor{blue!10}\emph{Improvement(\%)} 
                        & \cellcolor{blue!10}\textcolor{green!60!black}{0.00} & \cellcolor{blue!10}\textcolor{green!60!black}{+2.80} & \cellcolor{blue!10}\textcolor{green!60!black}{+4.89} & \cellcolor{blue!10}\textcolor{green!60!black}{+1.99} & \cellcolor{blue!10}\textcolor{green!60!black}{+1.08} & \cellcolor{blue!10}\textcolor{green!60!black}{+0.45} & \cellcolor{blue!10}\textcolor{green!60!black}{+14.59} & \cellcolor{blue!10}\textcolor{green!60!black}{+7.63} & \cellcolor{blue!10}\textcolor{red!80!black}{-1.84} & \cellcolor{blue!10}\textcolor{green!60!black}{+5.71} & \cellcolor{blue!10}\textcolor{green!60!black}{+5.54} & \cellcolor{blue!10}\textcolor{green!60!black}{+4.70} \\  
       \midrule
       \multirow{7}{*}{\rotatebox{90}{Claude4.5-Haiku}}
           & Zero-Shot & 0.6867 & 	0.3681 &	0.4372 &	0.3828 &  0.7100 &	0.3930 &	0.3770 &	0.3161  &0.4667	&0.1973&0.5096&0.4401\\
           & Random    & \underline{0.9900} & 	0.6669 &	0.1084 &	0.1815 &  \underline{0.9867} &	0.6040 &	0.1432 &	0.1754  &  \underline{0.9467} &	\underline{0.5138} &	0.2215&	0.2266 \\  
           & LDSIM     & \textbf{0.9933} & 	\textbf{0.7090} &	0.1027 &	0.1772 &  \textbf{0.9900} &	0.6159 &	0.1369 &	0.1707 & 0.9333&	0.5097&	0.2197&	0.2247  \\
           & Diversity & 0.9867 & 	0.6475 &	0.1187 &	0.1862 &  0.9767 &	0.5876 &	0.1451 &	0.1770 & \textbf{0.9533}&	0.5033&	0.2103 &	\underline{0.2160}  \\
            & BM25 &   \underline{0.9900} &\underline{0.7085}& \underline{0.0992} & \underline{0.1758} &0.9833& \underline{0.6310} &\underline{0.1348}&\underline{0.1645}&0.9400& 0.5110 & \underline{0.2045} & 0.2171 \\
           &\cellcolor{blue!10}\textbf{\mname{}~(Ours)} 
                        &\cellcolor{blue!10}\textbf{0.9933} & 	\cellcolor{blue!10}0.7039 &	\cellcolor{blue!10}\textbf{0.0975} &	\cellcolor{blue!10}\textbf{0.1757} &  \cellcolor{blue!10}\textbf{0.9900} &\cellcolor{blue!10}\textbf{0.6590 }&\cellcolor{blue!10}\textbf{0.1156} &\cellcolor{blue!10}\textbf{0.1554} & \cellcolor{blue!10}0.9433& \cellcolor{blue!10}\textbf{0.5253}& \cellcolor{blue!10}\textbf{0.2037 }&\cellcolor{blue!10}\textbf{0.2108}\\
           &\cellcolor{blue!10}\emph{Improvement(\%)} 
                        &\cellcolor{blue!10}\textcolor{green!60!black}{0.00} & 	\cellcolor{blue!10}\textcolor{red!80!black}{-0.72}&	\cellcolor{blue!10}\textcolor{green!60!black}{+1.71} &	\cellcolor{blue!10}\textcolor{green!60!black}{+0.06} &  \cellcolor{blue!10}\textcolor{green!60!black}{0.00} &	\cellcolor{blue!10}\textcolor{green!60!black}{+4.44} &	\cellcolor{blue!10}\textcolor{green!60!black}{+14.24} &	\cellcolor{blue!10}\textcolor{green!60!black}{+5.53} & \cellcolor{blue!10}\textcolor{red!80!black}{-1.05} & \cellcolor{blue!10}\textcolor{green!60!black}{+2.24} & \cellcolor{blue!10}\textcolor{green!60!black}{+0.39} & \cellcolor{blue!10}\textcolor{green!60!black}{+2.41} \\
           \bottomrule

    \end{tabular}
    \begin{tablenotes}[flushleft]   
		\item[1] \footnotesize \textcolor{green}{$\bm{\uparrow}$}: Higher value indicates better performance (VSR/IoU); \quad \textcolor{red}{$\bm{\downarrow}$}: Lower value indicates better performance (CD/ECD).
        \item[2] \footnotesize \textbf{Bold} values denote the best performance, and \underline{underlined} values denote the second-best performance.
        \item[3] \footnotesize \emph{Improvement(\%)} represents the performance improvement percentage of \mname{} over the the best result excluding \mname{} itself.
	\end{tablenotes}
    \end{threeparttable}
    }
\end{table*}

\paragraph{Optimal alignment}
In practice, simply applying an identical normalization to $\textbf{P}$ does not guarantee optimal alignment, although it can achieve an optimal result over IoU \cite{doris2025cad}. Thus, we consider a combinatorial space of operations for $\textbf{P}$:
\begin{itemize}
    \item \textbf{Translation \& Scaling:}
    We consider two reference choices for normalization: self-normalization based on $\mathbf{P}$ and direct alignment based on $\mathbf{Q}$.
    Specifically, we generate four normalized variants by taking the Cartesian product of translation and scaling choices:
    \[
    \mathcal{N}_{\mathbf{x}, r}(\mathbf{P})
    =
    \operatorname{Scale}\bigl(
        \operatorname{Translate}(\mathbf{P}, \mathbf{x}), r
    \bigr)\,,
    \]
    where $ \mathbf{x} \in \{-\bar{\mathbf{p}}, -\bar{\mathbf{q}}\}$, $r \in \{R_g(\mathbf{P}), R_g(\mathbf{Q})\}$

    \item \textbf{Rotation:}
    We consider only orthogonal rotations along the $\mathbf{XYZ}$ axes, with angles in $\{0^\circ, 90^\circ, 180^\circ, 270^\circ\}$, yielding 24 possible orientations. We exhaustively search this discrete rotation set to compare models under their best orthogonal alignment.
\end{itemize}

By evaluating these $4 \times 24 = 96$ candidate combinations, the optimal transformation $\Theta^*$ is determined by:
\begin{equation}
\Theta^* = \arg\max_{\Theta \in \Omega} \text{Metric}( \Theta(\textbf{P}), \textbf{Norm}(\textbf{Q}))
\end{equation}
where $\Omega$ denotes the set of 96 transformation candidates, the transformation that yields the highest similarity score, which considers both volumetric alignment (IoU) and surface fidelity (CD), is retained as the optimal alignment for evaluation. Example of different transformation is shown in \cref{fig:aligned_res}.




%

\subsection{Quantitative Analysis}
\begin{researchquestion}
How does DST compare to baseline strategies across different LLMs?
\end{researchquestion}
As shown in \cref{tab:baselines}, DST achieves the best overall performance across Qwen3-30B-A3B, DeepSeek-V3, and Claude 4.5-Haiku.
All ICL methods outperform the zero-shot baseline, confirming that exemplars provide useful CAD-specific knowledge.
However, random sampling is unstable, similarity-based methods such as LDSIM and BM25 may retrieve redundant exemplars, and the diversity baseline does not explicitly optimize query requirement coverage.
By contrast, DST selects complementary exemplars that maximize component-level coverage, leading to more consistent gains across metrics.

The improvement is especially clear on Qwen3-30B-A3B.
This is because Qwen3-30B-A3B has relatively weaker CAD generation capability, leaving more room for exemplar-based guidance.
In this setting, insufficient or redundant demonstrations can easily lead to missing geometric components or incorrect modeling operations, while DST compensates for this limitation by providing broader requirement-level knowledge coverage.
For stronger models such as DeepSeek-V3 and Claude 4.5-Haiku, VSR is already close to saturation, but DST still improves geometric quality, particularly in IoU, CD, and ECD.
Overall, these results demonstrate that knowledge sufficiency is a more effective exemplar selection objective than similarity, randomness, or point-wise diversity.

\begin{researchquestion}
How does task difficulty affect DST results?
\end{researchquestion}
As shown in \cref{tab:baselines}, DST improves performance across different difficulty levels, but its benefits are most pronounced on \textbf{Middle} tasks.
For Qwen3-30B-A3B, DST achieves the largest gains on Middle tasks, improving VSR and IoU by 20.79\% and 33.06\%, while reducing CD and ECD by 30.37\% and 25.95\%, respectively.
A similar trend can be observed on stronger models. DST reduces CD by 14.59\% on DeepSeek-V3 and 14.24\% on Claude 4.5-Haiku for Middle tasks, both larger than their improvements on Easy and Hard tasks.

This suggests that the benefit of DST is not simply proportional to task difficulty.
Easy tasks leave limited room for improvement, as their specifications are relatively simple and can already be handled well by LLMs or baseline retrieval methods.
Hard tasks, although still improved by DST, often involve complex spatial relations and many modeling operations, making it difficult for LLMs to fully absorb the retrieved knowledge.
In contrast, Middle tasks are complex enough to require external demonstrations but still tractable for LLMs to effectively reuse the complementary knowledge provided by DST.
Thus, DST is most effective when insufficient requirement coverage is the main bottleneck.

\begin{researchquestion}
How does model performance scale with the number of in-context exemplars?
\end{researchquestion}
As shown in \cref{fig:perto}, increasing the number of in-context exemplars generally improves performance for all selection strategies, confirming that additional demonstrations provide useful CAD-specific guidance.
However, the improvement is not linear: performance gains become smaller when the number of shots exceeds two, suggesting that task-relevant knowledge gradually approaches saturation.
Compared with baselines, DST reaches strong performance with fewer exemplars and approaches the Pareto front more efficiently.
For example, with only 3 shots, DST achieves an IoU comparable to or better than baselines using 5 shots.
This indicates that DST makes more effective use of the limited context budget by selecting complementary exemplars rather than simply adding more similar ones.
In contrast, similarity-based methods may introduce redundant demonstrations, causing additional shots to contribute limited new information.

\begin{researchquestion}
How does tiling ratio relate to performance?
\end{researchquestion}
As shown in \cref{fig:correlation}, DST consistently achieves a higher tiling ratio than all baselines as the number of shots increases.
Its tiling ratio steadily grows from approximately 0.3 to 0.5 across 1--5 shots, while baseline methods exhibit much slower growth.
This indicates that DST can more effectively expand the coverage of query-specific design components, rather than merely adding additional exemplars.

More importantly, the tiling ratio shows a clear positive correlation with generation performance.
A higher tiling ratio generally corresponds to higher VSR and IoU, as well as lower CD and ECD.
This validates tiling ratio as an effective surrogate for knowledge sufficiency: exemplars that cover more query components provide more useful guidance for generating syntactically valid and geometrically accurate CAD code.
The trend also explains why DST benefits more from additional shots than similarity-based methods, since each selected exemplar tends to contribute complementary information instead of redundant content.
Overall, these results confirm that improving component-level coverage is a key factor behind the performance gains of DST.

\begin{figure*}[!t]
    \centering
    \includegraphics[width=0.98\linewidth]{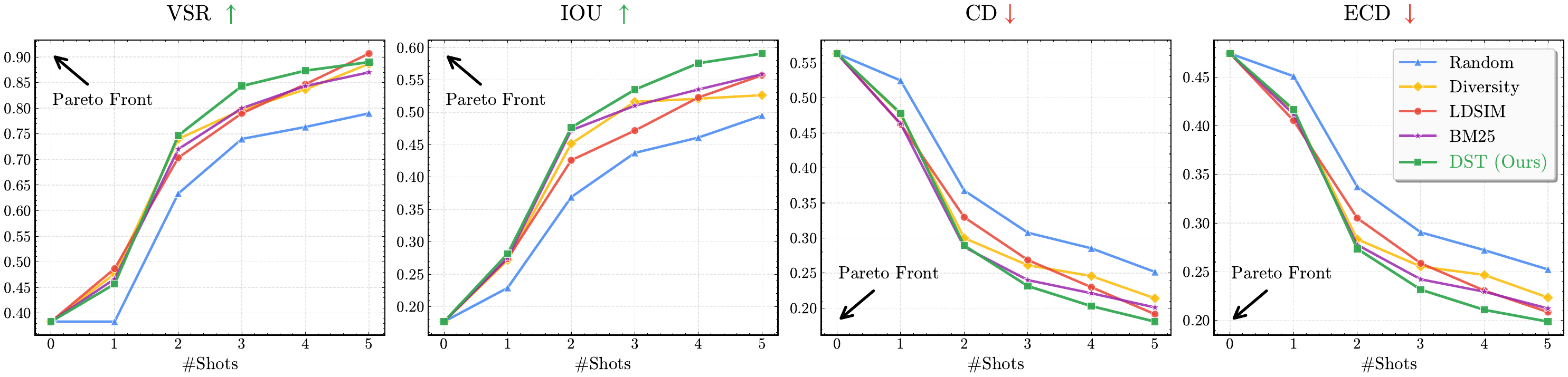}
    \caption{Performance trends of ICL strategies with increasing shot numbers based on DeepSeek-V3 on Hard tasks (Pareto front analysis)}
    \label{fig:perto}
\end{figure*}

\begin{figure*}[!t]
    \centering
    \includegraphics[width=0.98\linewidth]{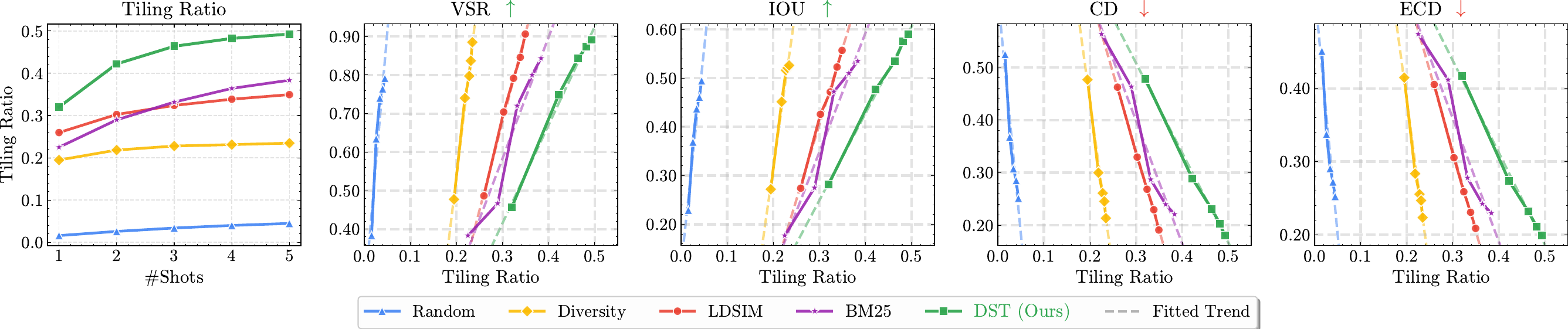}
    \caption{Correlation between Tiling Ratio, Shot Count, and Performance on DeepSeek-V3 on Hard tasks across different ICL strategies}
    \label{fig:correlation}
\end{figure*}

\begin{table*}[t]
    \centering
    \small
    \setlength\tabcolsep{5.2pt}
    \caption{\label{tab:softmatch}
    	Analysis of Soft and Hard Matching on Qwen3-30B-A3B (5-shot).
	}
    \resizebox{\textwidth}{!}{
    \begin{threeparttable}
    \begin{tabular}{l|cccccccccccc}
       \toprule
       \multirow{2}{*}{ {Strategy}} & 
       \multicolumn{4}{c}{ {Easy}} & \multicolumn{4}{c}{ {Middle}} & \multicolumn{4}{c}{ {Hard}} \\ 
       \cmidrule(lr){2-5}  \cmidrule(lr){6-9}  \cmidrule(lr){10-13}  
       &  {VSR}\textcolor{green}{$\bm{\uparrow}$} &  {IoU}\textcolor{green}{$\bm{\uparrow}$} &  {CD}\textcolor{red}{$\bm{\downarrow}$}  &  {ECD}\textcolor{red}{$\bm{\downarrow}$} &  {VSR}\textcolor{green}{$\bm{\uparrow}$} & {IoU}\textcolor{green}{$\bm{\uparrow}$} &  {CD}\textcolor{red}{$\bm{\downarrow}$}  &  {ECD}\textcolor{red}{$\bm{\downarrow}$} &  {VSR}\textcolor{green}{$\bm{\uparrow}$} &  {IoU}\textcolor{green}{$\bm{\uparrow}$} &  {CD}\textcolor{red}{$\bm{\downarrow}$}  &  {ECD}\textcolor{red}{$\bm{\downarrow}$} \\
        \midrule
        CodeBERT & \textbf{0.8567} & 0.6217 & \textbf{0.2626} & \textbf{0.2683} & 0.7633 & 0.4983 & 0.3090 & 0.2826 & 0.6700 & 0.3994 & 0.3437 & 0.3168 \\
        DST (Soft Matching) & \textbf{0.8567} & \textbf{0.6234} & 0.2659 & 0.2725 & 0.7167 &  0.4844 & 0.3453 & 0.3087 & 0.6433 & 0.4065 & 0.3628 & 0.3297 \\
        DST (Hard Matching)& 0.8067 & 0.6225 &  0.3021 & 0.2980 & \textbf{0.8133} & \textbf{ 0.5613} & \textbf{0.2760 }& \textbf{0.2511} & \textbf{0.6833} & \textbf{0.4260} & \textbf{0.3301} & \textbf{0.3044 } \\
        \bottomrule
        \end{tabular}
        \begin{tablenotes}[flushleft]
        \item[1] \footnotesize \textbf{Soft Matching}: two phrases are regarded as matched if the cosine similarity between their CodeBERT embeddings exceeds $0.98$. 
        \item[2] \textbf{Hard Matching}: strict lexical exact matching.
        \end{tablenotes}
        \end{threeparttable}
    }
\end{table*}

\begin{researchquestion}
Whether synonymous descriptions limit the performance of hard matching?
\end{researchquestion}
Hard matching strictly relies on literal string consistency, which is inherently constrained by synonymous paraphrasing in open-domain texts. In general scenarios, diverse synonymous descriptions prevent hard matching from identifying semantically equivalent phrases, thereby deteriorating its overall performance. Thus, We extend the method to Soft matching to evaluated the performance, which is designed to relieve this limitation by capturing implicit synonyms through embedding similarity (two phrases are regarded as matched if the cosine similarity between their CodeBERT embeddings exceeds 0.98). 
Nevertheless, as illustrated in Table~\ref{tab:softmatch}, such limitation does not dominate in this domain-specific dataset. The corpus consists of standardized technical terms with negligible synonyms and few alternative descriptions, eliminating the primary drawback of hard matching. Without the interference of synonymous variation, hard matching achieves precise and credible string matching. By comparison, soft matching loses its inherent advantage for synonym mining. The redundant embedding calculation introduces subtle numerical noise, which negatively affects the matching judgement. The experimental results verify that hard matching is limited by synonyms in open-domain data, while such constraints vanish on datasets with standardized terminology. In this case, soft matching cannot bring extra benefits and even degrades performance.

\subsection{Qualitative Analysis~(Case Study)}

\begin{researchquestion}
How do exemplars selected by DST differ from baseline strategies?
\end{researchquestion}
\cref{fig:case1} illustrates exemplar selection by DST and the best baseline LDSIM for CAD design specification.
For the first specification with two key features (``cylinder with hole'' and ``protruding rod''), LDSIM captures only ``cylinder with hole'' while missing ``protruding rod.''
For the second specification with three key features (``circular hole,'' ``oval shape,'' and ``tray''), LDSIM captures ``hole'' and ``oval'' but misses ``tray.''
In contrast, DST selects exemplars that tiles all features in both cases.
This demonstrates that the similarity top-$k$ approach~(LDSIM) retrieves homogeneous exemplars with incomplete coverage, while DST selects complementary exemplars that provide sufficient compositional knowledge.
By reducing overlap among selected exemplars, DST provides more diverse and complementary guidance to help the model better capture the compositional structure.


\begin{researchquestion}
How does generation quality change with increasing exemplar shots?
\end{researchquestion}
\Cref{fig:case2} compares generated CAD objects from DST and the best baseline LDSIM across four tasks with increasing exemplar shots.
DST consistently generates models closely matching ground truth in shape, topology, and features.
In contrast, LDSIM outputs exhibit structural flaws: fragmentation, missing holes, and extraneous geometry.
As exemplar shot increases, DST-generated objects outputs progressively converge toward the ground truth, while LDSIM shows instability: in three left three cases, outputs degrade with additional exemplars due to redundant selections.
This occurs because LDSIM's similarity-based selection retrieves redundant exemplars that fail to provide novel guidance, while DST consistently selects complementary exemplars that incrementally expand knowledge coverage and improve generation quality.

\begin{figure*}[!t]
    \centering
    \includegraphics[width=\linewidth]{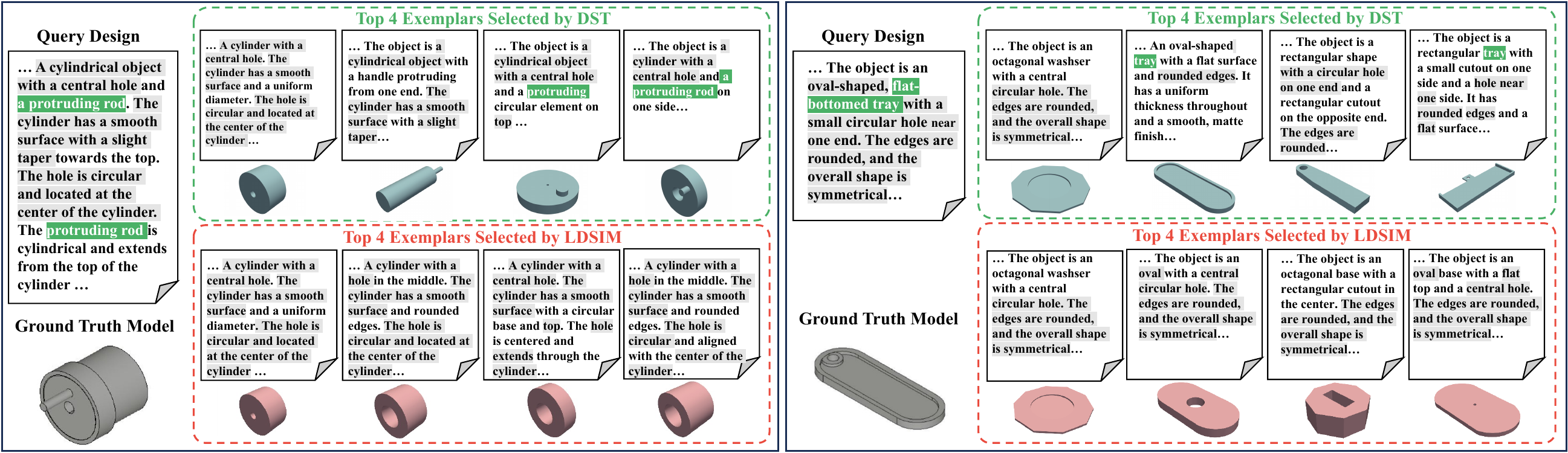}
    \caption{Exemplar Selection Quality Comparison Between DST and LDSIM for CAD Code Generation}
    \label{fig:case1}
\end{figure*}

\begin{figure*}[!t]
    \centering
    \includegraphics[width=\linewidth]{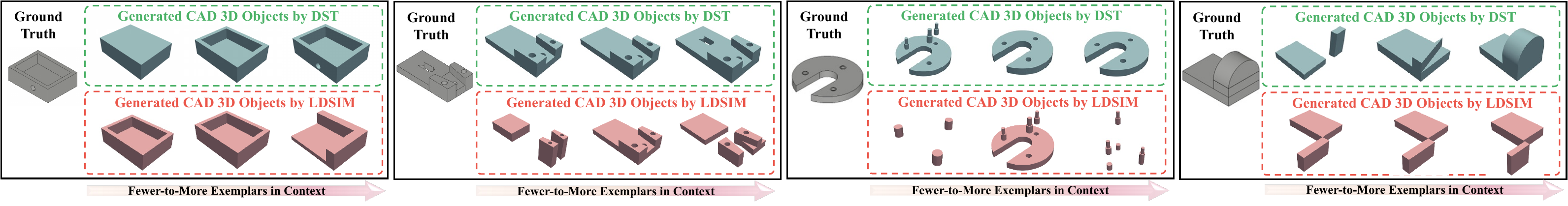}
    \caption{Generated CAD Object Quality Comparison of DST vs. LDSIM with Fewer-to-More In-Context Exemplars}
    \label{fig:case2}
\end{figure*}




\section{Related Works}\label{sec:rel_work}
\textbf{CAD Code Generation.}
With the advancement of generative AI, CAD code generation has emerged as a hot research field~\cite{niu2025cad,guan2025cad,alrashedy2025generating,niu2025cme,preintner2025evocad,willis2021fusion}. Its core goal is to generate precise, editable 3D models directly from design specifications, outperforming traditional 3D reconstruction methods by enabling parametric control and design flexibility~\cite{koch2019abc}. Existing approaches employ diverse modalities and technical designs:
1)~\textbf{Img2CAD}: Img2CAD takes 2D images as input and leverages vision large models to directly map visual content into CAD modeling instructions and executable code~\cite{wan2024cad,ahmed20253dfromllm,yuan2024openecad,daareyni2025generative,zhao2024chatcad+}. For instance, Xu et al.,~\cite{xu2024cad}, Alam et al.,~\cite{alam2024gencad} used vision-language models map image inputs to CAD commands; Wang et al.,~\cite{wang2025cad} introduced flexible code formats to enhance modeling expressiveness. 
2)~\textbf{PointClouds2CAD}: PointClouds2CAD adopts 3D point clouds as input, converting them into Python-based CAD code that can reconstruct the 3D CAD model after execution~\cite{dai2026meshcoder,ma2024draw}. For isntance, CAD-Recode translates a point cloud into Python codethat, when executed, reconstructs the CAD model~\cite{rukhovich2025cad}. 
3)~\textbf{Text2CAD}:
Compared to Img2CAD, Text2CAD has emerged as a prominent recent trend, owing to the fact that formal visual input is not only more challenging to acquire but also less generalizable than natural language design intent~\cite{govindarajan2025cadmium,zhou2026cad,liao2025automated}.
For instance, Khan et al.,~\cite{khan2024text2cad} used encoder-decoder architectures to translate textual descriptions into command sequences.
Guan et al.,~\cite{guan2025cad} leveraged CadQuery to improve the precision of LLM-based CAD code generation by integrating chain-of-thought reasoning and geometric reward mechanisms. 

\textbf{LLMs for Code Generation.}
With the rapid advancement of deep learning technologies, the field of intelligent code understanding~\cite{ma2023capturing,du2023pre,chai2022pyramid,du2025capturing} and code generation~\cite{xin2025enhancing,sun2026aces,du2024joint,du2023beyond} has witnessed explosive development.
Together, these advances have significantly promoted the automation of software engineering and expanded the application scope of large language models in programming-related scenarios, including code generation, automated testing, interactive development, and broader software engineering workflows~\cite{fan2023large, mastropaolo2023robustness, cassano2023multiple, schafer2024empirical,bouzenia2025repairagent,parasaram2025fact,yang2025revisiting}.
Traditional code generation tasks primarily tackled by seq-to-seq models optimized via supervised fine-tuning~\cite{svyatkovskiy2020intellicode,le2022coderl,li2022competition,zheng2023outline} suffered from limited generalization and high computational overhead for parameter updates as the model parameters increased.
The rise of LLMs has reshaped code generation paradigms, expanding real-world applications from function-level code generation~\cite{fakhoury2024llm} to full-stack project scaffolding~\cite{zhang2024codeagent}. 
The research efforts of code generation focusing on training-free strategies to boost performance, such as chain-of-thought~\cite{li2025structured,yang2024chain}, code ranking~\cite{zhang2023coder}, in-context learning~\cite{li2025large}, combining compiler testing~\cite{fakhoury2024llm}, etc. 
These strategies have pushed the performance of advanced code LLMs (e.g., CodeLlama~\cite{roziere2023code}, DeepSeek-Coder~\cite{guo2024deepseek}) well beyond that of traditional methods.

\section{Conclusions and Future Directions}\label{sec:conclusion}
This work addresses the underperformance of LLMs in domain-specific CAD code generation by introducing knowledge sufficiency as a principled objective for ICL exemplar selection. To optimize the objective, we propose DST, a submodular framework that quantifies component coverage via a tiling ratio, and leverages greedy optimization with a $(1-1/e)$ approximation guarantee. Extensive experiments demonstrate that DST consistently outperforms existing strategies across multiple LLMs and task difficulties, validating knowledge sufficiency as an effective objective. 
For future work, we aim to extend DST to other compositional generation tasks, enrich the tiling framework with learned semantic similarity, adapt it to hierarchical program structures, and further investigate how LLMs integrate compositional knowledge to develop structured spatial perception.

\clearpage

\bibliographystyle{ieeetr}
\bibliography{cad}

\clearpage

\twocolumn[
  \begin{@twocolumnfalse}
    \begin{center}
      \fontsize{18}{20}\selectfont\bfseries
          Supplementary Material \\ Design-Specification Tiling for ICL-based CAD Code Generation
      \vspace{1em}
    \end{center}
  \end{@twocolumnfalse}
]

\setcounter{page}{1}

\section{Inference-Time Prompt Template}\label{app:prompt}
At inference time, selected exemplars and their corresponding CAD code are structured into an in-context learning prompt for the language model. The prompt comprises three hierarchical components:

\begin{enumerate}
    \item \textbf{System prompt:} Establishes the model's role as a CAD code generation expert.
    \item \textbf{Task instruction:} Specifies output format constraints and generation requirements.
    \item \textbf{Demonstration sequence:} Presents $k$ exemplar pairs $\{(X_i, Y_i)\}_{i=1}^k$ ordered by relevance to the query.
\end{enumerate}

The query specification $X_{\text{query}}$ is appended as the final input, prompting the model to generate the corresponding CAD code $\hat{Y}$. The complete prompt structure follows:
\begin{tcolorbox}[
    colback=gray!5!white,
    colframe=gray!75!black,
    title=ICL Prompt Template,
    fonttitle=\bfseries
]

\small
\textbf{System Prompt:} \\
\texttt{You are an expert CAD engineer proficient in Python and the CadQuery library. Your task is to generate precise, executable CadQuery scripts according to given natural language descriptions.}

\vspace{1em}
\textbf{Instruction:} \\
\texttt{Please create a CadQuery Python code which can generate a model based on the instruction and description. \\
The final CadQuery code MUST BE put in '''python code''' with ONLY the executable code inside the python box, nothing else. \\
Relevant examples will be provided in sequence according to their similarity to the final query, and these examples may be helpful for answer generation. \\
Please don't use the non-existent '.scale()' method on Workplane objects.}

\vspace{1em}
\texttt{\#Examples Begin:}

\texttt{Description:} \{Design specification of the first exemplar\}

\texttt{'''} \{Code of the first example\} \texttt{'''} 

\vspace{1em}
$\cdots \cdots$
\vspace{1em}

\texttt{\#Examples End}

\vspace{1em}
\textbf{User Input:} \\
\texttt{Description:} \{Design specification of the current query\}
\end{tcolorbox}






\section{Baseline ICL Strategies}\label{app:baseline}
We evaluate several baseline in-context learning (ICL) exemplar selection strategies to benchmark our proposed DST selection strategy. These baselines include similarity-based approaches (e.g., LDSIM and BM25) as well as diversity-oriented techniques (e.g., Diversity).
The primary challenge in ICL is selecting an effective context for a given input source code $X_{\text{query}}$.
This problem can be formalized as selecting a subset of indices $S$ that satisfies the context size constraint $|S| = k$:
\begin{equation} \label{generaleq}
    S^* = \mathop{\mathrm{argmax}}_{S \subseteq \mathbb{N}_n, |S|=k} \mathrm{Score}(S,X_{\text{query}})\,.
\end{equation}
Each exemplar selection method can be viewed as a specific instantiation of the general formulation in Eq.~\ref{generaleq}.

\paragraph{(1) Random Sampling (Random)}
$k$ exemplars are uniformly sampled from the candidate pool, independent of the test input $X_{\text{test}}$:
\begin{equation}
S^* \sim \text{Uniform}(\mathbb{N}_n), \quad |S^*| = k\,.
\end{equation}

\paragraph{(2) Retrieval by Levenshtein Distance (LDSIM)}
Levenshtein distance (LD), also known as edit distance, is a fundamental string metric that quantifies similarity by counting the minimum number of single-character edit operations required to transform one sequence into another.
Due to its simplicity and low computational overhead, LD has been widely adopted in prior ICL literature for exemplar retrieval and matching tasks~\cite{shen2025magesql,yali_tse25_cast,liu2025towards}.
The selection score is defined as the negative sum of Levenshtein distances between each exemplar and the test input:
\begin{equation}
S^* = \mathop{\mathrm{argmax}}_{S \subseteq \mathbb{N}_n, |S|=k} -\sum_{i \in S} \text{LD}(X_i, X_{\text{test}})\,.
\end{equation}

\paragraph{(3) Retrieval by Best Matching 25 (BM25)}
BM25 is a bag-of-words retrieval function that is widely used in ICL~\cite{bm25llm} to rank textual documents by their relevance to a given query.
As a successor to the classic term frequency--inverse document frequency (TF--IDF) algorithm, BM25 addresses key limitations of TF--IDF, such as term frequency saturation, and introduces tunable parameters to adapt to diverse retrieval scenarios.
Consequently, BM25 has become a foundational baseline for retrieval tasks.
The overall score is computed as the sum of BM25 relevance scores between each candidate and the test input:
\begin{equation}
S^* = \arg\max_{S \subseteq \mathbb{N}_n,\, |S| = k} \sum_{i \in S} \text{BM25}(X_i, X_{\text{test}})\,.
\end{equation}

\paragraph{(4) Retrieval by Diversity with Clustering (Diversity)}
Following prior work~\cite{xiao2025role}, we implement a purely diversity-based strategy that partitions the training set into $k$ clusters using the k-means algorithm and selects the sample closest to each cluster centroid.
Cosine similarity computed over CodeBERT embeddings~\cite{codebert} is used as the distance metric:
\begin{equation}
\mathbb{N}_n \xrightarrow[]{\text{k-means}(f_{\text{CB}}(X_i))} \{C_1, \dots, C_k\}\,,
\end{equation}
where $C_j$ denotes the index set of the $j$-th cluster.
\begin{equation}
S^*_{[i]} = \arg\max_{j \in C_i} \cos\left(f_{\text{CB}}(X_j), f_{\text{CB}}(X_{\text{test}})\right)\,,
\end{equation}
for $i \in \{1,\dots, k\}$, and the final exemplar set is given by $S^* = S^*_{[1:k]}$.

\section{Additional Experiments}
\begin{figure}[!h]
    \centering

        \centering
        \includegraphics[width=\linewidth]{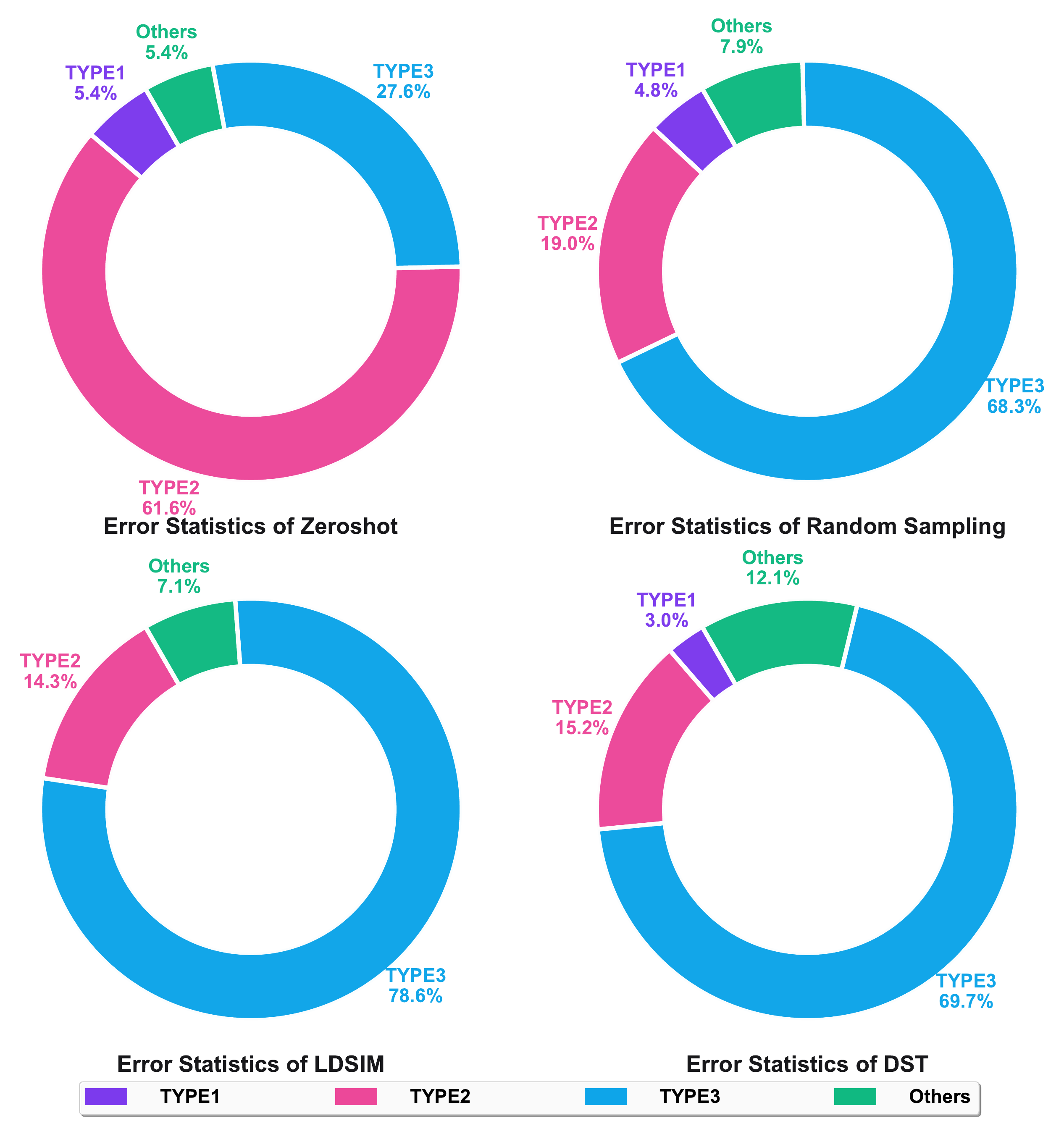}
    \caption{Failures distribution of DeepSeek-V3: The table (bottom) shows the count of each error types, while the pie charts (top) visualize the percentage of Error Type under Zeroshot, Random, LDSIM, and DST settings.}
    \label{fig:combined_fail_analysis}
\end{figure}

\subsection{Failure Analysis}
Throughout the pipeline \{source code extraction, compilation and execution, 3D rendering\}, the encountered failures can be systematically classified into the following three categories:
\begin{itemize}
    \item  \textbf{Type I: Response and Syntax Errors.} These failures primarily occur during the initial code extraction phase. Such errors typically stem from inconsistent input/output formats or malformed code blocks that fail to adhere to the required syntax.
    \item \textbf{Type II: API Misuse and Semantic Errors.} This category involves the invocation of non-existent attributes or the passing of invalid arguments to CadQuery methods. These issues reflect the model’s deficiency in domain-specific knowledge regarding the precise implementation details of third-party libraries.
    \item \textbf{Type III: Geometric Constraint Violations.} These are logical failures where the generated code violates fundamental geometric rules, leading to runtime crashes. Representative cases include: No pending wires present (i.e., attempting to create a solid from a non-contiguous path), Null-Entity Operations (i.e., extruding an empty point set), Non-coplanar operations (i.e., performing union or subtraction on workplanes that do not share the same plane).
\end{itemize}

\begin{table}[!t]
    \centering
    \small
    \setlength{\tabcolsep}{3pt} 
    \caption{\label{tab:fail_analysis}Failures distribution of DeepSeek-V3: The table shows the count of each error types under Zeroshot, Random, LDSIM, and DST settings(5-shot).}
    
    \begin{tabular}{lcccc}
    \toprule
    \textbf{Failure Type} & \textbf{Zeroshot} & \textbf{Random} & \textbf{LDSIM} & \textbf{DST} \\
    \midrule
    \textbf{TYPE I: Syntax Error}    & 10  & 3   & 0   & 1   \\
    \textbf{TYPE II: API Misuse}   & 114 & 12  & 4   & 5   \\
    \textbf{TYPE III: Geometric Error} & 51  & 43  & 22  & 23  \\
    \textbf{Others}        & 10  & 5   & 2   & 4   \\
    \midrule
    \textbf{Total Counts}         & 185 & 63  & 28  & 33  \\
    \bottomrule
    \end{tabular}
\end{table}


Taking DeepSeek-V3 as a representative case, the distribution of these error types across different input configurations is illustrated in \cref{fig:combined_fail_analysis}. 

Through a comparative analysis of the error distributions across the four experimental groups, two key insights emerge regarding the relationship between in-context learning and failure distribution: first the significant reduction in \textbf{Type II }(API Misuse) errors in the three few-shots groups (compared to the Zeroshot baseline) suggests that providing sufficient exemplars can largely rectify the model's misuse of third-party libraries.
Despite the overall improvement in success rates provided by LDSIM and DST, these methods remain insufficient for resolving \textbf{Type III} (Geometric Error). This suggests that Type III errors typically arise from intricate combinatorial logic and spatial reasoning that cannot be easily acquired through partially similar examples.


\subsection{More Case Study}
As illustrated in \cref{fig:morecase}, more case study was conducted among different strategies using three representative LLMs: Qwen3-A3B, DeepSeek-V3, and Claude4.5-Haiku. 

The top row presents the ground truth, with the rows below corresponding to the 3D CAD object generation outputs of the respective models under various strategies. As demonstrated in \cref{fig:morecase}, among all in-context learning (ICL) strategies evaluated, the Dialogue State Tracking (DST) approach consistently outperforms other methods across the three representative LLMs.

\begin{figure*}[!t]
    \centering
    \includegraphics[width=0.9\linewidth]{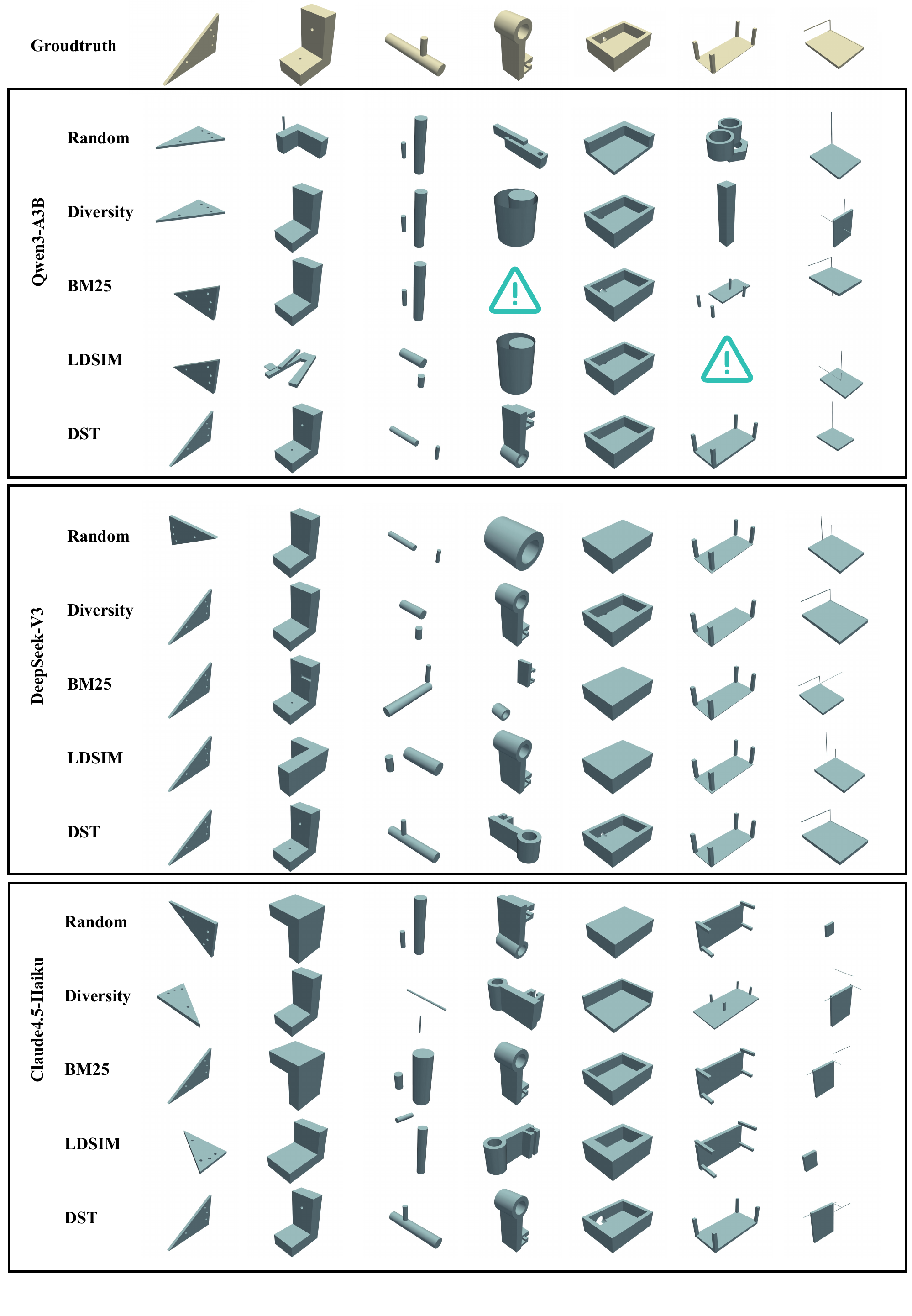}
    \caption{
        3D CAD objects generated by All three Models using different methods, where exclamation marks indicate the CAD code unable to produce valid CAD objects.
    }
    \label{fig:morecase}
\end{figure*}

\end{document}